\documentclass[sigconf]{aamas} 

\usepackage{balance} 

\setcopyright{ifaamas}
\acmConference[AAMAS '25]{This is a complete version of the paper 
(with the same title) that will appear in
Proc.\@ of the 24th International Conference
on Autonomous Agents and Multiagent Systems (AAMAS 2025)}{May 19 -- 23, 2025}
{Detroit, Michigan, USA}{A.~El~Fallah~Seghrouchni, Y.~Vorobeychik, S.~Das, A.~Nowe (eds.)}
\copyrightyear{2025}
\acmYear{2025}
\acmDOI{}
\acmPrice{}
\acmISBN{}

\newcommand{\calb}{\mbox{$\mathcal{B}$}}
\newcommand{\calc}{\mbox{$\mathcal{C}$}}
\newcommand{\calcp}{\mbox{$\mathcal{C}'$}}
\newcommand{\cald}{\mbox{$\mathcal{D}$}}
\newcommand{\calf}{\mbox{$\mathcal{F}$}}
\newcommand{\call}{\mbox{$\mathcal{L}$}}
\newcommand{\calp}{\mbox{$\mathcal{P}$}}
\newcommand{\cals}{\mbox{$\mathcal{S}$}}
\newcommand{\calsp}{\mbox{$\mathcal{S}'$}}

\newcommand{\psp}{\mbox{\calp(\cals)}}

\newcommand{\cnp}{\textbf{NP}}

\newcommand{\mc}[1]{\mathcal{#1}}

\newcommand{\bbb}{\mbox{$\mathbb{B}$}}

\newcommand{\tand}{\texttt{AND}}
\newcommand{\tor}{\texttt{OR}}
\newcommand{\txor}{\texttt{XOR}}

\newtheorem{theorem}{Theorem}[section]
\newtheorem{proposition}[theorem]{Proposition}
\newtheorem{lemma}[theorem]{Lemma}
\newtheorem{corollary}[theorem]{Corollary}

\usepackage[framemethod=tikz]{mdframed}
\newmdenv[backgroundcolor=gray!10, topline=false, bottomline=false, rightline=false, innerlinewidth=0.35pt, linecolor=black, roundcorner=2pt,innerleftmargin=2pt,
innerrightmargin=1pt
]{mybox-grey}

\usepackage{enumitem}

\acmSubmissionID{\texttt{1153}}

\title[Multi-Agent Systems Over Multilayer Networks]
{On Some Fundamental Problems for Multi-Agent Systems
     Over Multilayer Networks}

\author{Daniel J. Rosenkrantz}
\affiliation{
   \institution{University of Virginia} 
   \city{Charlottesville, VA}
   \country{USA}}
   \email{drosenkrantz@gmail.com}

\author{Madhav V. Marathe}
\affiliation{
   \institution{University of Virginia} 
   \city{Charlottesville, VA}
   \country{USA}}
   \email{marathe@virginia.edu}

\author{Zirou Qiu}
\affiliation{
   \institution{University of Virginia} 
   \city{Charlottesville, VA}
   \country{USA}}
   \email{zq5au@virginia.edu}

\author{S.~S. Ravi}
\affiliation{
   \institution{University of Virginia} 
   \city{Charlottesville, VA}
   \country{USA}}
   \email{ssravi0@gmail.com}

\author{Richard E. Stearns}
\affiliation{
   \institution{University of Virginia} 
   \city{Charlottesville, VA}
   \country{USA}}
   \email{thestearns2@gmail.com}

\begin{abstract}
Many researchers have considered multi-agent systems over single-layer
networks as models for studying diffusion phenomena.  Since real-world
networks involve connections between agents with different semantics
(e.g., family member, friend, colleague), the study of multi-agent
systems over multilayer networks has assumed importance.  Our focus
is on one class of multi-agent system models over multilayer networks,
namely  multilayer synchronous dynamical systems (MSyDSs).  We study
several fundamental problems for this model.  We establish properties
of the phase spaces of MSyDSs and bring out interesting differences
between single-layer and multilayer dynamical systems.  We show
that, in general, the problem of determining whether two given
MSyDSs are inequivalent is \cnp-complete.  This hardness result
holds even when the only difference between the two systems is the
local function at just one node in one layer.  We also present
efficient algorithms for the equivalence problem  for restricted
versions of MSyDSs (e.g., systems where each local function is a
bounded-threshold function, systems where the number of layers is
fixed and each local function is symmetric).  In addition, we
investigate the expressive power of MSyDSs based on the number of
layers.  In particular, we examine conditions under which a system
with $k \geq 2$ layers has an equivalent system with $k-1$ or fewer
layers.
\end{abstract}


\newcommand{\BibTeX}{\rm B\kern-.05em{\sc i\kern-.025em b}\kern-.08em\TeX}

\begin{document}


\pagestyle{fancy}
\fancyhead{}


\maketitle 

\baselineskip=\normalbaselineskip

\section{Introduction}
\label{sec:intro}

\subsection{Background and Motivation}

Motivated by the spread of epidemics, information and social behavior in populations,
the topic of contagion propagation over networks
has been studied extensively in the literature
(see e.g., \cite{EK-2010,granovetter1978threshold,watts2002simple,wasserman1994social,newman2018networks}).
When the contagion represents information about
a product produced by an organization, the propagation
models enable the organization to develop
strategies for maximizing the spread of 
the information (see e.g., \cite{kempe2003maximizing}).
When the contagion is an epidemic, the propagation
models enable public
health officials to take appropriate actions to contain the spread of the epidemic 
(see e.g., \cite{adiga2022ai,lefevere2022ai}).
Such studies  typically use networked
dynamical systems models, where nodes
represent individuals and edges represent relationships
that enable a contagion to spread in the
underlying network
(see e.g., \cite{adiga2019graphical,Adiga-etal-wsc-2019,Barrett-etal-2011}).
Each node has an active or inactive state 
which indicates whether or not the node has acquired
the contagion.
The \textbf{configuration} of the system at any time
is the vector of state values of nodes at that time.
The interaction between a node and its neighbors is
based on an appropriate local function, and this interaction
may change the state of the node.
Throughout this paper, we will assume that nodes compute
and update their states \emph{synchronously}.
The time evolution of the system is represented by
a sequence of successive configurations.
The global behavior of a dynamical system is captured by 
its \textbf{phase space}, which is a directed graph in which
each node is a configuration and each directed edge indicates
a single step transition from one configuration to another.

This work has generally assumed that the underlying network
consists of a single layer; that is, there is only one type of 
relationship between entities for contagion propagation.
A number of fundamental questions have been studied
over the years for synchronous dynamical systems  
over single-layer networks. Following the literature
(see e.g., \cite{adiga2019graphical}), we refer to these
single-layer systems as SyDSs.
Papers on SyDSs have addressed questions concerning reachability
of configurations (i.e., does a given SyDS
starting from a configuration \calc{} reach a 
configuration \calcp{}?), equivalence (i.e., are two given 
SyDSs equivalent in terms of their phase space?),
existence of \textbf{fixed points} (i.e., configurations in which none
of the nodes changes its state), etc.; see e.g., 
\cite{Adiga-etal-wsc-2019,OU-2017,Chistikov-etal-2020,Rosenkrantz-etal-2015,rosenkrantz2018testing}
and the references cited therein.

A number of papers have pointed out that single-layer
networks are inadequate to capture complex real-world
contagion phenomena.
The reason is that single-layer networks model
just one type of relationship
among individuals.
Instead, one needs to consider 
\underline{multilayer}\footnote{In much of the literature,
the word ``multilayer'' is not hyphenated.
We follow the same practice here.} 
networks (also known as
\underline{multiplex} networks  
\cite{zhang2016survey,kivela2014multilayer,porter-2018})
that capture multi-relational structures
(e.g., family member, friend, colleague, contact in a social network)
among individuals to effectively model contagion propagation
\cite{kivela2014multilayer,hammoud2020multilayer,newman2018networks}.
As a simple example, information that an individual obtains through
a social network such as Facebook or Twitter may play
a role in deciding whether the individual decides to wear a mask or get
vaccinated.
Thus, it is important to study dynamical systems where the
underlying graph has multiple layers.
In standard models of multilayer networks, there is
just one set of nodes; each layer may have a different
set of edges modeling different types of 
relationships (see, e.g., 
\cite{de2013mathematical,de2016physics,porter-2018,del2016spreading,newman2018networks}).
The model for a dynamical system over a multilayer network
is a simple extension of the corresponding model over a
single-layer network.
In such a system, each node has a \emph{single} state value 
but a possibly different 
interaction function for each layer.
At each time instant, after computing the value of the interaction function of the node in each layer, a \textbf{master function} 
uses these values to compute the next state of that node.
As mentioned earlier,
we focus on the \emph{synchronous} update scheme where all the
nodes compute their next states and update their states
simultaneously.
For consistency with previous work \cite{Qiu-etal-icml-2024},
we refer to this model as a \textbf{multilayer synchronous dynamical
system} (MSyDS).
The notions of  configuration and phase space for a MSyDS are
the same as
those for a dynamical system over a single-layer network.
Algorithmic work on multilayer networks has centered around 
defining and
computing appropriate centrality measures 
(see e.g., \cite{de2013mathematical,porter-2018,newman2018networks}).
There has also been some work on estimating the sizes of cascades
under certain contagion models
(e.g., susceptible-infected-recovered or SIR model) through analysis and simulation in multilayer networks
(see e.g., \cite{del2016spreading,newman2018networks} and
the references cited therein).
Additional discussion regarding work on multilayer systems
appears in Section~\ref{sse:related}.
For multilayer dynamical systems, there has been some
work \cite{Qiu-etal-icml-2024} on learning the local interaction functions of such 
systems under the probably approximately correct (PAC)
model \cite{valiant1984theory}.
However, to our knowledge, fundamental issues such as the structure of phase spaces and the equivalence of multilayer dynamical systems
have not been addressed in the literature.

\subsection{Contributions} 
\label{sse:contrib}
Our work takes the first step towards 
exploring some fundamental
aspects of multilayer synchronous dynamical systems (MSyDSs).
A summary of our contributions is provided below.

\begin{enumerate}[leftmargin=*,noitemsep,topsep=0pt] 
\item \underline{\textsf{Phase space properties:}}
We present MSyDSs
whose local functions are  threshold functions but whose phase spaces include long cycles in contrast to
the phase spaces of single-layer SyDSs.
For example, we show (see Section~\ref{sec:phase}) that there
is a MSyDS with two layers such that
all its local functions are threshold functions and
its phase space contains a cycle whose length is exponential
in the number of nodes. 
We also give another construction of a MSyDS with $n$ nodes
and $2n-1$ layers such that all local functions are threshold
functions and the phase space is a cycle 
of length $2^n$.
In contrast, it is known that for single-layer SyDSs where
each local function is a threshold  function, the maximum
length of a cycle in the phase is 
two~\cite{Goles-Martinez-2013}.
Our constructions of MSyDSs with threshold local functions also show that their phase spaces
may contain simple directed paths whose length is exponential in the number of nodes.
(These paths are part of an exponentially long cycle.)
For single-layer SyDSs with $n$ nodes and threshold local functions, it is known that the length of any simple path in the phase space is $O(n^2)$~\cite{BH+06,Goles-Martinez-2013}.

\item \underline{\textsf{Equivalence Problem MSyDSs:}}~ 
Given two MSyDSs \cals{} and \calsp{} on the same set of nodes, 
we say that they are equivalent if their phase spaces (as directed graphs)
are identical.
We show (see Section~\ref{sec:complexity}) that the 
equivalence problem for MSyDSs is \cnp-hard even 
for highly restricted versions of the MSyDSs.
Specifically, our reduction from 3SAT (see \cite{GJ-1979} for
a definition) produces two MSyDSs \cals{} and \calsp{}
over a set of nodes with the following properties:
(i) in each layer, the graphs of \cals{} and \calsp{}
are identical; (ii) the graph in each layer is a star
graph with zero or more isolated nodes;
(iii) all local functions are threshold 
functions; and (iv) the MSyDSs \cals{} and \calsp{}
differ in only the threshold of one node in one layer.
Thus, a minor difference between two MSyDSs is sufficient to
obtain computational intractability for the equivalence
problem for MSyDSs with threshold local functions.
In contrast, it is known that for two single-layer SyDSs
where the underlying graphs are the same and the local
functions are threshold functions, the equivalence
problem can be solved efficiently \cite{Adiga-etal-wsc-2019}.

\item \underline{\textsf{Efficient Algorithms for the Equivalence Problem for}}\\
\underline{\textsf{Restricted Classes of MSyDSs:}}~
We also provide efficient algorithms 
(see Section~\ref{sec:special}) for the equivalence
problem for several special classes of MSyDSs.
These special classes include the following: 
(i) the number of layers in the two MSyDSs is \emph{fixed} 
and each local function is 
symmetric\footnote{Symmetric functions (which properly contain threshold 
functions) are defined in Section~\ref{sec:prelim}.}; 
(ii) each local function is a threshold function, 
the maximum node degree in each layer is \emph{bounded}
and each master function is \texttt{OR} or each master
function is \texttt{AND}.
These algorithms are obtained by reducing 
the equivalence problem for the two systems to a
form of equivalence with respect to each node and 
showing that the node equivalence problem can be solved
efficiently.

\item \underline{\textsf{Expressive Power of MSyDSs:}}~ 
We initiate the study of the expressive power
of MSyDSs (see Section~\ref{sec:expressive}) based on the number of layers.
We observe that for any MSyDS over $k \geq 2$ layers, there is
an equivalent single-layer SyDS with more complex 
local functions.
However, if one also restricts the class of
local functions, we show that there are MSyDSs with $k$
layers for which there is no equivalent MSyDS with fewer
layers.
\end{enumerate}

\noindent
We also observe that many analysis problems (e.g., existence
of fixed points) that are computationally intractable for
single-layer dynamical systems remain so for multilayer
systems as well.

\subsection{Additional Remarks}
\label{sse:addl_remarks}

Long cycles (and paths) in the phase space are indicative
of the complexity of the \textbf{reachability problem} 
for the underlying
dynamical system (i.e., given a multilayer SyDS \cals{} and configurations
\calc{} and \calcp{}, can \cals{} reach \calcp{} 
starting from \calc{}?). As mentioned in
Section~\ref{sse:contrib},  in any single layer system 
with $n$ nodes and threshold
local functions, the length of any phase space path is $O(n^2)$, and
any such path ends in a limit cycle of length at most 2.  So, the
reachability problem for such a system is efficiently solvable.  However,
when the phase space of a system contains exponentially long
cycles/paths, the reachability problem is likely to be computationally
intractable; such complexity results have appeared in several references
(e.g., \cite{BH+06,Chistikov-etal-2020}).
The novelty in our work is that with just two
layers, one can construct multilayer systems with threshold functions
whose phase spaces contain exponentially long cycles (and hence
exponentially long paths). To our knowledge, this is the first work
that formally establishes phase space properties of multilayer systems.

Our hardness result for the equivalence problem for multilayer
systems with threshold functions holds even when there is a
very minor difference between the two systems.  
(As will be seen from the proof in Section~\ref{sec:complexity}, the
only difference between the two systems constructed through a reduction
from the 3SAT problem is in
the threshold value of just one node in one layer.)
We highlight this to point out that while 
the equivalence problem for single layer systems with threshold
functions is efficiently solvable~\cite{Adiga-etal-wsc-2019},
a very minor difference 
suffices to make the problem computationally intractable
for multilayer systems.

Finally, we note that when the number of layers is \emph{fixed},
our efficient algorithm for the Equivalence problem for multilayer 
SyDSs allows \emph{symmetric} local functions which properly contain 
threshold functions.

\subsection{Related Work}
\label{sse:related}

Networked dynamical systems provide a systematic
formal framework for agent-based models (ABM) and
to capture interactions among agents in a network. 
Wellman~\cite{wellman-2016-putting}
discusses the relationships between ABM and multi-agent systems.
Networked dynamical system models have been used by the
multi-agent systems community to study a variety of topics,
including contagion propagation, graphical games and migration due to
catastrophic events (see e.g., 
\cite{li2018modelling,jiang2014diffusion,kearns-2007,mehrab2024network}).
Researchers have also studied formal
aspects of various computational problems for networked
dynamical systems (see e.g., 
\cite{BH+06,Chistikov-etal-2020,Rosenkrantz-etal-2015,rosenkrantz2018testing}).
While this work has been in the context of single-layer systems,
various research issues in the context of multilayer networks
are also being actively pursued by the research community
(see \cite{zhang2016survey} and the references cited therein).
For example, a number of papers have studied structural properties of multilayer networks and
the computation of various centrality measures 
for such networks (see \cite{zhang2016survey,chakraborty2016cross,mittal2017mining} and
the references cited therein).
In addition, there is work on other topics on multilayer networks,
including
reliability issues, analysis of cascades,
enabling cooperation, and applications in various
domains (see e.g., 
\cite{jiang2015reliable,li2015cross,kivela2014multilayer,del2016spreading,hammoud2020multilayer,Stojkoski-etal-2018,qiu2022understanding}).

To our knowledge, the only reference which considers fundamental
computational questions regarding multilayer dynamical systems is
\cite{Qiu-etal-icml-2024}.
As mentioned earlier, this work considers learning the 
local functions of MSyDSs under the PAC model but does not
address the research questions studied in this paper.


\section{Preliminaries}
\label{sec:prelim}

\subsection{Multilayer Synchronous Dynamical Systems}
\label{sse:msyds}

In this section, we present the definitions associated
with multi-layer synchronous dynamical systems.
The notation and terminology used in this section
are based on the presentation 
in~\cite{Qiu-etal-icml-2024}.

Unless otherwise mentioned, the networks considered in our
paper are undirected.
We begin with the definition and notation for multilayer networks.
A \textbf{multilayer network}~\cite{kivela2014multilayer} with $k \geq 1$
layers is a set of graphs $M = \{G_i: 1 \leq i \leq k\}$, 
where $G_i = (V, E_i)$ is the graph in the $i^{\text{th}}$ layer.
Thus, all of the graphs have the same node set $V$ with $n$ nodes, but the
edge sets in different layers may be different.
The definition of a discrete dynamical system over a multilayer network
is a generalization of the corresponding definition for
a single-layer system which has been studied by many
researchers (see e.g., \cite{MR-2007,BH+06,Adiga-etal-wsc-2019}).
Our focus is on networked dynamical systems over the
domain \bbb{} = \{0,1\}.
A \textbf{multilayer synchronous dynamical system} (MSyDS) \cals{}
over the domain \bbb{} has the following components.

\begin{description}[leftmargin=*,noitemsep,topsep=0pt]
\item{(a)} A multilayer network $M = \{G_i(V, E_i) : 1 \leq i \leq k\}$
with $k$ layers.
Each  node $v \in V$ has a \textbf{state} from \bbb{}.

\item{(b)} A collection \calf{} = $\{f_{i,v} : 1 \leq i \leq k,~ v \in V\}$
of functions, with $f_{i,v}$ denoting the \textbf{local interaction function}
for node $v$ in layer $i$.

\item{(c)} A collection $\Psi = \{\psi_v : v \in V\}$ of functions,
with $\psi_v$ denoting the \textbf{master function} of node $v$.
\end{description}

Let $V = \{v_1, v_2, \ldots, v_n\}$.
At any time, the \textbf{configuration} \calc{} of the
system is an $n$-vector $(s_1, s_2, \ldots, s_n)$, where $s_i \in \bbb$
is the state of node $v_i$ at that time, $1 \leq i \leq n$.
Given the configuration at time $t$, the configuration of the system
at time $t+1$ is computed as follows.

\begin{enumerate}[leftmargin=*,noitemsep,topsep=0pt]
\item In each layer $i$, each node $v$ computes the value of its
local function $f_{i,v}$.
The inputs to the function $f_{i,v}$ are the state of $v$ and those of its neighbors in $G_i$, and the output of $f_{i,v}$ is a value in \bbb{}.
We use $W_v(i)$ to denote this  output value.
Since there are $k$ layers, this stage provides a vector $\mathbf{W}_v$ 
of $k$ values
$(W_v(1), W_v(2), \ldots, W_v(k))$ for each node $v \in V$.

\item Then, for each node $v$, its master function $\psi_v$ is evaluated.
The input to $\psi_v$ is $\mathbf{W}_v$ and the output of $\psi_v$
is a value in \bbb{}.
This value becomes the state of $v$ in time step $t+1$.

\item All the nodes carry out the above computations and update
their states \emph{synchronously}.
\end{enumerate}

Let $\calc_t$ and $\calc_{t+1}$ denote the configurations of a MSyDS \cals{} at times $t$ and $t+1$ respectively.
We refer to $\calc_{t+1}$ as the \textbf{successor} 
of $\calc_t$ and 
$\calc_{t}$ as the \textbf{predecessor} of $\calc_{t+1}$.
Since we restrict our attention to deterministic local
and master functions, each configuration of a MSyDS has a \emph{unique} successor;
however, a configuration may have zero 
or more predecessors \cite{Qiu-etal-icml-2024}.

\smallskip 

\noindent
\textbf{A Note About the Graph Model:}~
In real-world applications, multilayer systems
may have a different set of nodes and edges in each layer. When
this happens, one can merge all the node sets into a single set and
have different sets of edges in each layer. Hence, the commonly used
model for multilayer systems 
(e.g., \cite{de2013mathematical,de2016physics,hammoud2020multilayer,jiang2015reliable, kivela2014multilayer}) 
assumes one set of nodes for all the layers but (possibly) 
different sets of edges in different layers. 
Our results also use this single node set model.

\smallskip

\noindent
\textbf{Graph Theoretic Terminology:}~ Let $M$ denote the
underlying multilayer graph of a MSyDS \cals{}. 
Recall that $G_i(V,E_i)$ denotes the
graph in layer $i$ of $M$.
The \textbf{degree} of a node $v$ in layer $i$
is the number of neighbors of $v$ in that layer.
For a node $v$ in layer $i$, the 
\textbf{closed neighborhood} of $v$ in that layer includes $v$
and each node $u$ such that $\{u,v\} \in E_i$.
To discuss the structure of phase spaces of networked
dynamical systems, we will also need some terminology for directed
graphs.
In a directed graph, the \textbf{outdegree} of a node $x$
is the number of outgoing edges from $x$.
The phase space may also contain
\textbf{self loops}, that is, directed edges
of the form $(x,x)$.
A \textbf{simple path} in a directed graph is a directed path
in which all nodes are distinct.  
Similarly, a \textbf{simple cycle} in a directed graph is a directed cycle
in which all nodes are distinct.

\smallskip 

\noindent
\textbf{Local and Master Functions:}~  
Since the domain under consideration is \bbb{} = \{0,1\},
we consider several classes of Boolean functions as local
and master functions. 
The definitions provided below are available in many standard
references such as \cite{Crama-Hammer-2011}.
An important class of Boolean functions studied in the literature
on the spread of social contagions is that of \textbf{threshold} 
functions (see e.g., 
\cite{granovetter1978threshold,watts2002simple,chen2021network}).
For any integer $\tau \geq 0$, the
$\tau$-\textbf{threshold} function has the
value 1 iff at least $\tau$ of its inputs 
have the value 1.
We note that the 1-threshold function is the \tor{} function.
Also, the \tand{} function with $q$ inputs is $q$-threshold function.
When the local function $f_{i,v}$ of a node $v$ is the 
$\tau$-threshold function, we say that its threshold condition
is satisfied if at least $\tau$ of the nodes in the
closed neighborhood of $v$ in $G_i$ are in state~1.
Symmetric Boolean functions are a superset of the class
of threshold functions.
The value of a \textbf{symmetric} Boolean function depends
only on the number of 1's in the input.
Thus, each threshold function is also a symmetric function.
Likewise, the \txor{} function is symmetric.
Each symmetric Boolean function with $q$ inputs
can be specified by a table with $q+1$ rows, with row $i$
specifying the value of the function when the number of 1's
in the input is exactly $i$, $0 \leq i \leq q$.
We refer to this as the \textbf{symmetry table} for
the (symmetric) function.
Our work considers threshold and symmetric Boolean functions 
for local functions as well as master functions.

\begin{figure}
    \centering
    \includegraphics[width=0.55\textwidth]{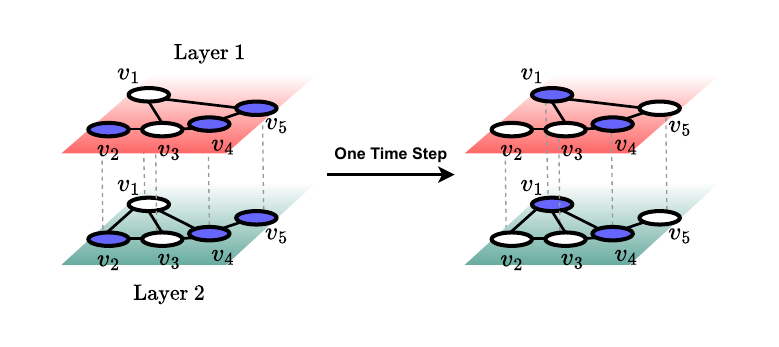}
    \vspace{-20px}
    \caption{An example of a small MSyDS, where state-1 nodes are
    highlighted in blue. The master function at each node is
    \texttt{OR}. The thresholds of the nodes (same for the two
    layers) are: $v_1$ : $2$, $v_2$ : $2$, $v_3$ : $3$, $v_4$ :
    $1$ and $v_5$ : $3$.}
    \label{fig:msyds_example}
\end{figure}

\smallskip 

\noindent
\textbf{An Example for a MSyDS:}~ To illustrate the above definitions,
we present a simple example of a 2-layer MSyDS in
Figure~\ref{fig:msyds_example}.  
One can think of the edges in each of the two layers as representing
different relationships (e.g., family member, friend) between pairs of nodes.
There are 5 nodes in the system
and the local function for each node in each layer is a threshold
function. These threshold values are indicated in the caption of
Figure~\ref{fig:msyds_example}.
For simplicity, we have chosen the same threshold
value for each node in the two layers.
The master function at each node is \tor{}.
The initial configuration of the system is shown
in the left panel of Figure~\ref{fig:msyds_example}),
where nodes in state-1 are highlighted in blue.
Thus, in the initial configuration, nodes $v_2$, $v_4$ and $v_5$
are in state-1 while nodes $v_1$ and $v_3$
are in state~0. 
This configuration \calc{} is represented by $(0, 1, 0, 1, 1)$.

We now explain how the system transitions from one configuration
to the next. 
As mentioned above, let \calc{} = $(0, 1, 0, 1, 1)$ denote
the current configuration.
Consider node $v_1$ which is in state~0 in \calc{}.
In Layer~1, $v_1$ has only neighbor (namely $v_5$) in state~1.
In Layer~2, $v_1$ has two neighbors (namely $v_2$ and $v_4$) in state~1.
Since the threshold of $v_1$ is 2 in 
each layer, the local functions of $v_1$ in 
layers 1 and 2 have values 0 and 1 respectively.
Since the master function is \tor{}, the next
state of $v_1$ is 1.
In a similar manner, one can verify that the
next states of $v_2$, $v_3$, $v_4$ and $v_5$
are 0, 0, 1 and 0 respectively.
Thus, \calcp{}, the successor of \calc{},
is given by $(1, 0, 0, 1, 0)$, as shown in 
the right panel of Figure~\ref{fig:msyds_example}. 
It can also be verified that the successor of \calcp{}
is \calcp{} itself; that is, once the system reaches the 
configuration \calcp{} = $(1, 0, 0, 1, 0)$, no more state
changes occur, and the system stays in that configuration
for ever. Thus, \calcp{} is a \emph{fixed point}
of the system.  \hfill$\Box$

\smallskip 

\noindent
\textbf{Additional terminology regarding dynamical systems:}~
We review some necessary terminology regarding networked dynamical
systems from the literature 
(see e.g., \cite{MR-2007,BH+06,adiga2019graphical}).
Suppose the node set of the underlying multilayer graph of
a MSyDS \cals{} over \bbb{} has $n$ nodes.
Thus, the total number of distinct configurations is $2^n$.
The \textbf{phase space} of \cals{}, denoted by \psp{}, is directed graph with a node for each possible configuration;
for each pair of nodes $x$ and $y$, there is a directed
edge $(x,y)$ in \psp{} iff node $y$ represents the configuration
which is the successor of the configuration represented by 
the node $x$.
A \textbf{transient} in the phase space is a directed path
leading to a cycle.
Since our MSyDSs are deterministic, each configuration
has a unique successor. In other words, the
outdegree of each phase space node is 1.
Thus, the phase space has $2^n$ nodes 
and $2^n$ edges.

\subsection{Problem formulations} 
\label{sse:prob_form}

We study several fundamental questions concerning MSyDSs.
The first of these questions involves the structure of the phase
spaces of MSyDSs.
In particular, we study how the (directed) paths and cycles in 
the phase spaces of MSyDSs can be significantly different compared
to those of single-layer SyDSs, even when the local functions
used by the dynamical systems are threshold functions.
In fact, our results show that one can construct MSyDSs on two-layer networks
whose phase spaces have significantly longer cycles and paths
compared to single-layer SyDSs.

The second question involves the equivalence of two given
MSyDSs.
Let $\mc{S}$ and $\mc{S}'$ be a pair of MSyDSs with the same node set $V$.
We say that $\mc{S}$ and $\mc{S}'$ are {\em equivalent} 
if their phase spaces (considered as directed graphs) are identical.
We formulate the corresponding decision problem as the \emph{inequivalence} problem as follows.

\smallskip

\noindent
\textbf{Inequivalence of MSyDSs:}

\noindent
\underline{\textsf{Given:}} Two MSyDSs \cals{} and \calsp{}
on the same set $V$ of nodes.

\noindent
\underline{\textsf{Question:}} Are the phase spaces of \cals{}
and \calsp{} \emph{different}?

\smallskip 

The reason for considering the Inequivalence problem is that
the problem belongs to the class \cnp{}.
To see this,
define an \textbf{inequivalence witness} as a configuration \calc{} over $V$, such that the successors of $\calc$ under $\mc{S}$ and $\mc{S}'$ 
are different. 
The Inequivalence problem is in \cnp{} since one can guess an inequivalence witness \calc{}, and efficiently verify that the
successors of \calc{} under \cals{} and \calsp{}
are different. We show that the Inequivalence problem is
\cnp-complete even when there is only a minor difference
between the given MSyDSs.
We also present efficient algorithms for restricted
versions of the Inequivalence problem.

The last fundamental question that we address concerns 
the expressive power of MSyDSs in terms of the number of layers.
Specifically, we investigate whether for a given MSyDS with $k$
layers, one can obtain an equivalent MSyDS with fewer layers.
Our results point out that while this can be done if more complex
local functions are permitted, it may not be possible if there
are additional restrictions on the local functions.


\section{Phase Space Properties of Multilayer Systems}
\label{sec:phase}

In this section, we examine phase space properties of multilayer systems.
This examination identifies some interesting differences between the
phase space properties of single-layer and multilayer systems.
We begin with a result which points out that there are 
MSyDSs whose phase spaces contain long cycles.
This construction is also used in  
establishing another result 
(Theorem~\ref{thm:two_layers_exponential_cycles}) of this section.


\begin{theorem}\label{thm:unary-cycles}
For every $n \geq 1$, there is a MSyDS $\mc{S}_n$ with two layers and $n$ nodes, 
where every local function is a threshold function, and every master 
function is symmetric,
with the following properties:
the phase space of $\mc{S}_n$ contains exactly one cycle, 
the length of this cycle is $n+1$,
and every transient is of length one.
\end{theorem}

\begin{proof}
Let the nodes of $\mc{S}_n$ be denoted as $v_i$, $1 \leq i \leq n$.
The graph in both layers is a complete graph.
In layer 1, the threshold of each node is $n$.
In layer 2, the threshold of each node $v_i$, $1 \leq i \leq n$, is $i-1$.
Every master function is \texttt{XOR}, which is a symmetric function.

Now, consider the phase space of $\mc{S}_n$.
For each $j$, $0 \leq j \leq n$, let $\calc_j$ denote the configuration of $\mc{S}_n$
such that for each $v_i$, $1 \leq i \leq n$, if $i \leq j$ then $\calc_j[v_i] = 1$,
and if $i > j$ then $\calc_j[v_i] = 0$.
Note that the successor configuration of each 
$\calc_j$ is $\calc_{j+1 \, \mbox{mod} \, (n+1)}$.
Thus, the $n+1$ configurations $\calc_j$, $0 \leq j \leq n$, 
form a phase space cycle of length $n+1$.
Moreover, for any configuration $\calc$, let $|\calc|$ denote the number of 1's in $\calc$,
i.e., the Hamming weight of $\calc$. Let $j' = |\calc| + 1 \, \mbox{mod} \,(n+1)$.
Then, the successor configuration of $\calc$ is $\calc_{j'}$, 
a configuration in the above phase space cycle.
\end{proof}

The next result shows that the cycles in the phase spaces of MSyDSs with just two layers
can be exponentially large in the number of nodes
of the system.

\begin{theorem}\label{thm:two_layers_exponential_cycles}
For MSyDSs with two layers, threshold local functions, and \texttt{XOR} master functions, a phase space cycle can be exponentially large in the number of nodes.
\end{theorem}

\begin{proof}
For an integer $q \geq 1$, let $p_1, p_2, \ldots , p_q$ be the first $q$ primes, 
where we assume that 2 is the first prime. 
Let $\Sigma_q$ be the two-layer MSyDS obtained as the union of the $q$ two-layer MSyDSs 
$\mc{S}_{p_1-1}, \mc{S}_{p_2-1}, \ldots,  \mc{S}_{p_q-1}$, 
each of which is constructed as in Theorem~\ref{thm:unary-cycles}.
More precisely, the node set of $\Sigma_q$ is the union of the node sets of the $q$
constituent MSyDSs, 
the edge set in each layer is the union of the edge sets on that layer of the constituent MSyDSs, 
and each node has the same master function as in its constituent MSyDS.
Note that neither layer of $\Sigma_q$ has an edge between nodes occurring in two distinct constituent MSyDSs.
Since each constituent MsyDS $\mc{S}_{p_j-1}, 1 \leq j \leq q$,
has a phase space cycle of length $p_j$,
$\Sigma_q$ has a phase space cycle of length $\Pi_{j=1}^q p_j$.
In this construction, the number of nodes in the MSyDS is
$\sum_{j=1}^{q} p_j -q$.
It is known that $\sum_{j=1}^{q} p_j$ is
asymptotically $\Theta(q^2 \log{q})$ \cite{Sinha-2015}.
As explained above, the length of a cycle in the
phase space is $\prod_{j=1}^{q} p_j$,
which is asymptotically $\Omega(e^{q\log{q}})$ \cite{primorial-wiki}.
Thus, the length of this cycle is exponential in the number of nodes
in the system.
\end{proof}

The next result shows that there are MSyDSs where all the
nodes in the phase space form a simple cycle. 

\begin{theorem}\label{thm:exponential-cycles}
For every $n \geq 2$, there is a MSyDS $\mc{S}_n$ with $n$ nodes, 
for which every local function is a threshold function, 
and every master function is symmetric,
whose phase space consists of a single cycle, of length $2^n$.
\end{theorem}

\begin{proof}
Let the nodes $V$ of $\mc{S}_n$ be denoted as $v_i$, $1 \leq i \leq n$.
For any configuration $\calc$ of $V$, let $\widehat{\calc}$ denote 
the integer encoded by $\calc$, with $\calc[v_1]$ viewed as the low order bit.
Let $\calc'$ denote the successor configuration of $\calc$ under $\mc{S}_n$.
$\mc{S}_n$ will be constructed to have the property that $\widehat{\calc'}$ 
will equal $\widehat{\calc}+1 \, \mbox{mod} \, 2^n$.
This property ensures that the entire phase space of $\mc{S}_n$ 
consists of a single length $2^n$ cycle.

For any node $j$, $1 \leq j \leq n$, 
let $V_j$ denote the set of nodes $\{ v_1, v_2, \ldots , v_j \}$.
Note that $V_j$ corresponds to bit $j$ and the lower bits of $\widehat{\calc}$.

MSyDS $\mc{S}_n$ is constructed as follows.
There are $2n-1$ layers.
The layer 1 graph contains no edges.
For each $j$, $2 \leq j \leq n$, the two layers $j$ and $n+j-1$ contain the same graph.
This graph contains the $j-1$ edges $\{v_j,v_i\}$, $1 \leq i < j$,
i.e., an edge between $v_j$ and each node in $V_{j-1}$.

In layer 1, every node has threshold 1.
In layer $j$, $2 \leq j \leq n$, node $v_j$ has threshold $j-1$,                         
and every other local function is the constant function 0,
i.e. a threshold function with threshold equal to the node degree plus two.
In layer $n+j-1$, $2 \leq j \leq n$, node $v_j$ has threshold $j$,
and every other local function is the constant function 0.   

The master function for node $v_1$ is \texttt{NOR}.
For each node $v_j$, $2 \leq j \leq n$, 
the master function is the symmetric function that is 1 iff
either exactly one or exactly two of its $2n-1$ inputs equal 1.


We claim that for any configuration $\calc$ of $V$,
the successor configuration $\calc'$ has the required property that
$\widehat{\calc'} \, = \, \widehat{\calc}+1 \, \mbox{mod} \, 2^n$.
Note that bit 1 of $\widehat{\calc}+1$ is the complement of bit 1 of $\widehat{\calc}$.
Also note that for each $j$, $2 \leq j \leq n$, bit $j$ of $\widehat{\calc}+1$ equals 1 iff
either $\calc[v_j] = 1$ and at least one node in $V_j$ has value 0 in $\calc$,
or $\calc[v_j] = 0$ and all the nodes in $V_{j-1}$ (i.e., all the lower order bits) 
have value 1 in $\calc$.

Note that for node $v_1$,
the only layer whose local transition function for $v_1$ is not the constant 0 is layer 1; 
thus $\calc'[v_1] = \overline{\calc[v_1]}$.

Next, consider any other node $v_j$, $2 \leq j \leq n$.
Note that the only layers whose local transition function for $v_j$ is not the constant 0 are layers 1, $j$, and $n+j-1$.

Suppose that $\calc[v_j] = 1$.
Then, $f_{1,v_j}(\calc) = 1$,
i.e., the local function for node $v_j$ in layer 1 evaluates to 1. 
Also, at least one node in $V_j$ has value 0 in $\calc$ iff 
$f_{n+j-1,v_j}(\calc) = 0$, i.e.,
the local function for node $v_j$ in layer $n+j-1$ evaluates to 0.
Thus, when $\calc[v_j] = 1$, at least one node in $V_j$ has value 0 in $\calc$ 
iff either exactly one input ($f_{1,v_j}$)
or exactly two inputs ($f_{1,v_j}$ and $f_{j,v_j}$) 
to the master function for $v_j$ have value 1.
Consequently, when $\calc[v_j] = 1$, $\calc'[v_j]$ equals bit $j$ of $\widehat{\calc}+1$.

Now, suppose that $\calc[v_j] = 0$.
Then, $f_{1,v_j}(\calc)$ and $f_{n+j-1,v_j}(\calc)$ both equal 0. 
Moreover, since $\calc[v_j] = 0$, all the nodes in $V_{j-1}$ have value 1 in $\calc$ iff 
$f_{j,v_j}(\calc) = 1$.
Thus, when  $\calc[v_j] = 0$,  
if all the nodes in $V_{j-1}$ have value 1 in $\calc$, 
then exactly one input ($f_{j,v_j}$) to the master function for $v_j$ has value 1,
and if at least one node in $V_{j-1}$ has value 0 in $\calc$, 
then no input to the master function for $v_j$ has value 1.
Consequently, when $\calc[v_j] = 0$, $\calc'[v_j]$ equals bit $j$ of $\widehat{\calc}+1$. This completes the proof.
\end{proof}

As mentioned in Section~\ref{sec:intro}, the above results
point out that even with just two layers, one can design
MSyDSs with threshold local functions whose phase spaces 
contain cycles which are much larger than the ones possible
in single-layer SyDSs with threshold local functions.

\section{Complexity of Equivalence}
\label{sec:complexity}
In this section, 
we show that the inequivalence problem is \cnp{}-Complete, 
even when there is just a minor
difference between given pair of MSyDSs \cals{} and \calsp{}.

\begin{theorem}\label{thm:inequiv_NPcomplete}
The MSyDS inequivalence problem is \cnp{}-complete,
even when $(i)$ all local functions are threshold functions, 
$(ii)$ the two multilayer networks are identical,
$(iii)$ the graph in each layer is a star graph plus a set 
(possibly empty) of isolated nodes, 
$(iv)$ all master functions are \texttt{OR},
and $(v)$ the two given MSyDSs differ only in the value of the threshold of one node in one layer.
\end{theorem}


\begin{proof}
As mentioned in Section~\ref{sse:prob_form},
the MSyDS inequivalence problem in \cnp{}.
For \cnp{}-hardness, a reduction from the 3SAT problem \cite{GJ-1979} proceeds as follows.
Let $f$ be the given CNF formula. 
Let $n$ be the number of variables in $f$ and $m$ be the number of clauses.
The constructed MSyDSs $\mc{S}$ and $\mc{S}'$ have the same node set
$V$, consisting of $2n+1$ nodes, as follows.
For each variable $x_i \in X$,
$V$ contains the two nodes $y_i$ and $z_i$.
Intuitively, node $y_i$ corresponds to the literal $x_i$,
and node $z_i$ corresponds to the literal $\overline{x_i}$.
We refer to these $2n$ nodes as {\bf literal nodes}.
There is also one additional node, $w$, which we refer to as the {\bf center node}.

The constructed pair of multilayer graphs each contains node set $V$ and $n+m+1$ layers.
We refer to the set of layers as $\Gamma_1 \cup \Gamma_2 \cup \Gamma_3$.
$\Gamma_1$ consists of a single graph $G_1$.
$\Gamma_2$ consists of $n$ graphs, which we refer to as $G_2^i$, $1 \leq i \leq n$.
$\Gamma_3$ consists of $m$ graphs, which we refer to as $G_3^j$, $1 \leq j \leq m$. 

Graph $G_1$ contains $2n$ edges, 
and is a star graph with node $w$ in the center, and the other $2n$ nodes as leaves.
Each graph $G_2^i$ contains the two edges $\{w,y_i\}$ and $\{w,z_i\}$.
Each graph $G_3^j$ contains an edge for each literal in clause $c_j$.
The endpoints of a given edge are $w$ and the literal node 
for the \emph{complement} of the literal occurring in the clause.
For instance, if clause $c_j$ is $x_3 \vee x_6 \vee \overline{x_9}$,
then $G_3^j$ contains the three edges $\{w,z_3\}$, $\{w,z_6\}$ and $\{w,y_9\}$.

The threshold of every node other than node $w$ is 0 in every graph.
The threshold of node $w$ in each graph $G_2^i$ is 2.
The threshold of node $w$ in each graph $G_3^j$ is the number of literals in clause $j$.
The threshold of node $w$ in graph $G_1$ is $n$ in $\mc{S}$, and is $n+1$ in $\mc{S}'$.

This completes the construction of $\mc{S}$ and $\mc{S}'$.
Note that the only difference between $\mc{S}$ and $\mc{S}'$ is the value of the threshold of node $w$ in $G_1$.

On the one hand, suppose that $f$ is satisfiable.
Let $\alpha$ be a satisfying assignment for $f$.
Let $\calc_{\alpha}$ be the configuration of $V$ such that: 
$\calc_{\alpha}(w) =0$; and for each variable $x_i$ of $f$,
if $\alpha(x_i) =1$ then $\calc_{\alpha}(y_i) = 1$ and $\calc_{\alpha}(z_i) = 0$,
and if $\alpha(x_i) =0$ then $\calc_{\alpha}(y_i) = 0$ and $\calc_{\alpha}(z_i) = 1$.
It can be verified that in the successor configuration of $\calc_{\alpha}$,
the value of node $w$ is 1 under $\mc{S}$, and is 0 under $\mc{S}'$.
Thus, $\mc{S}$ and $\mc{S}'$ are inequivalent.

On the other hand, suppose that $\mc{S}$ and $\mc{S}'$ are inequivalent.
Then there exists a configuration $\calc$ of $V$, whose successor configuration 
under $\mc{S}$ and $\mc{S}'$ differ.

Let $\cald$ and $\cald'$ denote the successor of $\calc$ under $\mc{S}$ and $\mc{S}'$, respectively.
Since $\mc{S}$ and $\mc{S}'$ only differ in the threshold value of node $w$ 
in graph $G_1$, 
this difference in threshold is responsible for successor configurations $\cald$ and $\cald'$ being different.
Thus, the number of 1's in $\calc$ is exactly $n$, $\cald(w) = 1$ and $\cald'(w) = 0$.

Without loss of generality, we can assume that $n \geq 2$.
If $\calc(w)$ were to equal 1, 
then at least one of the graphs in $\Gamma_2$ would make $\cald'(w)$ equal 1;
thus, $\calc(w) = 0$.
Consequently, the $n$ nodes with value 1 in $\calc$ are all literal nodes.
If for some variable $x_i$ of $f$, the literal nodes $y_i$ and $z_i$ both have value 1 in $\calc$,
graph $G_2^i$ would make $\cald'(w)$ equal 1.
Thus, for each of the $n$ variables of $f$, 
exactly one of the two corresponding literal nodes has value 1 in $\calc$.

\newcommand{\subcalc}{\mathcal{C}}

Let $\alpha_{\subcalc}$ denote the assignment to the variables of $f$, 
such that for each variable $x_i$, $\alpha_{\subcalc}(x_i) = 1$ iff $\calc(y_i) = 1$.

Now consider any clause of $f$, say clause $c_j$.
If $\alpha_{\subcalc}$ does not satisfy clause $c_j$, 
then all the literal nodes that are neighbors of $w$ in graph $G_3^j$ would have value 1 in $\calc$,
thereby making $\cald'(w)$ equal 1.
Thus, assignment $\alpha_{\subcalc}$ satisfies all the clauses of $f$,
so $f$ is satisfiable.

Thus, $f$ is satisfiable iff $\mc{S}$ and $\mc{S}'$ are inequivalent. This concludes the proof.
\end{proof}

\noindent
\textbf{Remark.} The equivalence problem for single-layer
systems with threshold local functions is known to be 
efficiently solvable~\cite{Adiga-etal-wsc-2019}.
Thus, from the complexity perspective, 
the above theorem points out a significant difference between 
single-layer and multilayer systems with threshold local functions.

\section{Efficient Algorithms for Equivalence 
for Special Classes of MSyDSs}
\label{sec:special}

We first introduce several key concepts which  help us derive
efficient solutions for the equivalence problem for
several restricted classes of MSyDSs.
Let $\mc{S}$ and $\mc{S}'$ be a pair of MSyDSs with the same node
set $V$. For a given node $v \in V$, we say that $\mc{S}$ and
$\mc{S}'$ are $v$-\emph{equivalent} if for every configuration
$\calc$ of $V$, the value of $v$ in the successor of $\calc$ under
$\mc{S}$ and $\mc{S}'$ are the same.  Note that $\mc{S}$ and $\mc{S}'$
are equivalent iff for every node $v \in V$, $\mc{S}$ and $\mc{S}'$
are $v$-equivalent.  Thus, an efficient algorithm for node-equivalence
would lead to an efficient algorithm for MSyDS equivalence.

For a given threshold function, we use the term {\em negative
threshold} to mean the number of 0's that make the function equal
0.  Thus, for a given node $v$ of a MySDS, if the local function
for $v$ on a given layer is a threshold function with threshold
$t$, and the node has degree $\delta$ in that layer, the negative
threshold for that local function is $\delta - t + 2$.

\subsection{Fixed Number of Layers}

Suppose that for a node $v$ of a given pair of MSyDSs $\mc{S}_0$
and $\mc{S}_1$ with a common  node set $V$, we want to determine
the $v$-equivalence of $\mc{S}_0$ and $\mc{S}_1$.  We could, for
each of the $2^{|V|}$ configurations $\calc$ of $V$, compare the
value of $v$ in the successor configuration of $\calc$ under
$\mc{S}_0$ and under $\mc{S}_1$.  However, since the number of
configurations is exponential in the number of nodes, this
straightforward approach would take time exponential in the number
of nodes.

In contrast, we show that when there is a bound on the number of
layers of the given MSyDSs, and each local function is symmetric,
$v$-equivalence can be determined in time polynomial in the number
of nodes, where the degree of the polynomial depends on the bound
on the number of layers.  Our algorithm for doing this is based on
constructing a partition $\Pi_v$ of the node set $V$, such that $\Pi_v$
has the following two key properties.

\begin{enumerate}[leftmargin=*,noitemsep,topsep=0pt]
    \item The number of blocks in partition $\Pi_v$ is a function
    of the number of layers in the given MSyDSs, and is independent
    of the number of nodes.  
    \item For any configuration $\calc$ of $V$,
     the value of node $v$ in the successor configuration of $\calc$
     for each of the given MSyDSs is determined by how many of the
     nodes in each of the blocks of $\Pi_v$ equal 1 in $\calc$.
\end{enumerate}

For any configuration $\calc$ of $V$, let the $v$-\emph{profile} of $\calc$ be a count of how many of the nodes in each of the blocks of $\Pi_v$ are equal to 1 in $\calc$.
Since the number of blocks  of $\Pi_v$ is polynomial, the number of possible profiles is polynomial, and so the set of all possible $v$-profiles can be explored in polynomial time.
The approach of exploring all $v$-profiles yields a polynomial time algorithm, as described by the following result.

\begin{theorem}\label{thm:fixed_number_layers}
If the number of layers is bounded by a fixed value for both dynamical
systems, the MSyDS equivalence problem can be solved in polynomial
time when all local functions are symmetric functions, even when
each master function is arbitrary.  
\end{theorem} 

\begin{proof}
Let $K$ be the fixed bound on the number of layers.  Let $\mc{S}_0$
and $\mc{S}_1$ denote the given pair of MSyDSs, with common node
set $V$.  Let $k_0$ and $k_1$ be the number of layers in $\mc{S}_0$
and $\mc{S}_1$, respectively.  Let $k' = k_0 + k_1$.  We represent
a layer of $\mc{S}_0$ or $\mc{S}_1$ as a pair $(a,\ell)$, where $a
\in \calb = \{0,1\}$ specifies one of the two given MSyDSs, and
$\ell$ is a layer number of $\mc{S}_a$, i.e., $1 \leq \ell \leq
k_a$.  We refer to such a pair as an \emph{anchored-layer}.

Let $L$ denote the set of $k'$ anchored layers.  Let $\mc{L}$ denote
the power set of $L$, so that $|\mc{L}| = 2^{k'}$.  We define a
\emph{profile} to be a vector with an element for each member of
$\mc{L}$, where each element value is an integer in the range 0
thru $|V|$.  Thus, a profile has an element for each subset of $L$.

Now, consider a given node $v \in V$.  Let $\xi_v: V \to \mc{L}$
denote the function that maps each node $u$ into the set of anchored
layers $(a,\ell)$ such that $u$ and $v$ are 
generalized neighbors\footnote{We use the term ``generalized
neighbors'' to mean that $u$ and $v$ are in each other's
closed neighborhood.}
in layer $\ell$ of $\mc{S}_a$.  
Note that $\xi_v$ induces a partition
of $V$, with each block of the partition consisting of those nodes
that $\xi_v$ maps into the exact same set of anchored layers.  We
refer to this partition as $\Pi_v$.  Thus, a pair of nodes $u$ and
$w$ are in the same block of $\Pi_v$ iff $\xi_v(u) = \xi_v(w)$,
i.e., for each layer of $\mc{S}_0$ and each layer of $\mc{S}_1$,
$u$ and $w$ are either both generalized neighbors of $v$ in that
layer or neither is a generalized neighbor of $v$.

Based on partition $\Pi_v$, for each set $\lambda \in \mc{L}$, we
let $V_{v,\lambda}$ denote the set of nodes in the block of the
partition associated with $\lambda$, i.e., the set of nodes $u$
such that $\xi_v(u) = \lambda$.  Thus, $V_{v,\lambda}$ is the set
of nodes $u$ such that $v$ and $u$ are generalized neighbors in
every anchored layer in $\lambda$, and in no other anchored layer.

We say that a given profile $\theta$ is a  $v$-\emph{profile} if
for every $\lambda \in \mc{L}$, $\theta[\lambda] \leq |V_{v,\lambda}|$.
Let $\Theta_v$ denote the set of all $v$-profiles.

We now specify how a given $v$-profile $\theta$ specifies a value
for node $v$ for each of the two given MSyDSs $\mc{S}_0$ and
$\mc{S}_1$.  Let function $\rho_v: \Theta_v \times \calb \to \calb$
be defined as follows: Consider a given $(\theta,a) \in \Theta_v
\times \calb$.  For each anchored layer $(a, \ell)$ of $\mc{S}_a$,
let $$w_v((a,\ell)) = \sum\limits_{\lambda \in \mc{L} \, | \,
(a,\ell) \in \lambda} \theta[\lambda]$$ and let $W_v((a,\ell))$ be
the value of the symmetric local function $f_{\ell,v}$ of $\mc{S}_a$
when exactly $w_v((a,\ell))$ of its inputs equal 1.  This defines
a $k_a$-vector $\mathbf{W}_{v,a}$ of Boolean values, where
$\mathbf{W}_v(\ell) = w_v((a,\ell))$.  Then, $\rho_v(\theta,a)$ is
defined to be the value of the master function $\psi_{v,a}$ of
$\mc{S}_a$ when $\mathbf{W}_{v,a}$ is the input to $\psi_{v,a}$.

Let $\mu_v$ denote the function that maps each configuration $\calc$
over $V$ into a $v$-profile, as follows.  For each configuration
$\calc$ over $V$, $\mu_v(\calc)$ is the $v$-profile such that for
each $\lambda \in \mc{L}$, $\mu_v(\lambda)$ is the number of nodes
in $V_{v,\lambda}$ that have value 1 in $\calc$.  Note that $\mu$
is an onto function, i.e., for every $v$-profile $\theta$, there
is at least one configuration $\calc$ of $V$ such that $\mu_v(\calc)
= \theta$.  Such a configuration can be constructed by setting, for
each $\lambda \subseteq L$, exactly $\theta(\lambda)$ members of
$V_{v,\lambda}$ to value 1, and the other members of $V_{v,\lambda}$
to value 0.

Note that for each MSyDS $\mc{S}_a$ and any configuration $\calc$
of $V$, the value of node $v$ in the successor configuration of
$\calc$ under $\mc{S}_a$ equals the value of $\rho_v(\mu_v(\calc),a)$.
Thus, we can determine the $v$-equivalence of $\mc{S}_0$ and
$\mc{S}_1$ by comparing the values of $\rho(\theta,0)$ and
$\rho(\theta,1)$ for every $v$-profile $\theta$.

While the number of configurations is exponential in the number of
nodes, the number of elements in any profile is bounded by a function
of $K$, and the value of each component of any $v$-profile is at
most the number of nodes.  Since $K$ is fixed, the number of profiles
is polynomial in the number of nodes.  Thus, this approach yields
a polynomial time algorithm, as follows.

\smallskip

\noindent
\textbf{The algorithm.} We now describe a polynomial time algorithm
for equivalence, utilizing the above concepts.  Given $\mc{S}_0$
and $\mc{S}_1$, the algorithm first constructs $\call$.  Then, for
each node $v$, the algorithm determines whether $\mc{S}_0$ and
$\mc{S}_1$ are $v$-equivalent, as follows.  The algorithm constructs
partition $\Pi_v$, and then constructs $\Theta_v$, the set of all
$v$-profiles.  Next, for each $\theta \in \Theta_v$, the algorithm
computes $\rho(\theta,0)$ and $\rho(\theta,1)$.  Note that $\mc{S}_0$
and $\mc{S}_1$ are $v$-equivalent iff for every $\theta \in \Theta_v$,
$\rho(\theta,0) = \rho(\theta,1)$.

Finally, $\mc{S}_0$ and $\mc{S}_1$ are equivalent iff for every
node $v$ and every $\theta \in \Theta_v$, $\rho(\theta,0) =
\rho(\theta,1)$.  
\end{proof}


\subsection{Bounded Threshold}

We need a couple of lemmas to get an efficient algorithm
for this special case.

\begin{lemma}\label{lem:max_threshold}
Let $\mc{S}_0$ and $\mc{S}_1$ be a pair of MSyDSs with common node
set $V$, such that every local function is a threshold function,
and all master functions are \texttt{OR}.  Let $\tau$ denote the
largest threshold value occurring in $\mc{S}_0$ and $\mc{S}_1$.
Then $\mc{S}_0$ and $\mc{S}_1$ are inequivalent iff there is an
inequivalence witness with at most $\tau$ 1's.
\end{lemma}


\begin{proof}
Trivially, if there is an inequivalence witness with at most $\tau$ 1's, 
then $\mc{S}_0$ and $\mc{S}_1$ are inequivalent.

Now, suppose that $\mc{S}_0$ and $\mc{S}_1$ are inequivalent.
Let $\calc$ be an inequivalence witness.
Let $\cald_0$ and $\cald_1$ denote the successor configuration of $\calc$ 
under $\mc{S}_0$ and $\mc{S}_1$, respectively.
Let $v$ be a node such that $\cald_0[v] \neq \cald_1[v]$.
Wlog, assume that $\cald_0[v] = 0$ and $\cald_1[v] = 1$.
Then, there is at least one layer of $\mc{S}_1$ such that for $\calc$, 
the local function for $v$ in $\mc{S}_1$ evaluates to 1.
Let $\ell$ be any such layer. 
Let $t_v$ denote the threshold of $v$ in layer $\ell$ of $\mc{S}_1$.
Let $V'$ be any set of \emph{exactly} $t_v$ nodes that are all 1 in $\calc$ 
and are generalized neighbors of $v$ in layer $\ell$ of $\mc{S}_1$.

Let $\calc'$ denote the configuration in which the nodes in $V'$
have value $1$, and all other nodes have value $0$.  Let $\cald'_0$
and $\cald'_1$ denote the successor configuration of $\calc'$ under
$\mc{S}_0$ and $\mc{S}_1$, respectively.  Then, for $\calc'$, the
local function for $v$ in layer $\ell$ of $\mc{S}_1$ evaluates to
$1$, so $\cald_1'[v] = 1$.  Since every node with value $0$ in
$\calc$ still has value $0$ in $\calc'$, we have $\cald_0'[v] = 0$.
Thus, $\calc'$ is an inequivalence witness with at most $\tau$ 1's.
This concludes the proof.  
\end{proof}

We now note that a dual lemma holds, with regard to maximum negative
threshold value. The proof is the dual to the proof of
Lemma~\ref{lem:max_threshold}.


\begin{lemma}\label{lem:max_negative_threshold}
Let $\mc{S}_0$ and $\mc{S}_1$ be a pair of MSyDSs with common node
set $V$, such that every local function is a threshold function,
and all master functions are \texttt{AND}.  Let $\nu$ denote the
largest negative threshold value occurring in $\mc{S}_0$ and
$\mc{S}_1$.  Then $\mc{S}_0$ and $\mc{S}_1$ are inequivalent iff
there is an inequivalence witness with at most $\nu$ 0's.
\end{lemma}

We are now ready to prove the main result of this subsection.

\begin{theorem}\label{thm:bounded_threshold}
(a) If all local functions are threshold functions, the maximum
threshold value is bounded by a fixed value, and all master functions
are \texttt{OR}, the MSyDS equivalence problem can be solved in
polynomial time.

(b) If all local functions are threshold functions, the maximum
negative threshold value is bounded by a fixed value, and all master
functions are \texttt{AND}, the MSyDS equivalence problem can be
solved in polynomial time.  
\end{theorem} 

\begin{proof}
(a) Let $\tau_{max}$ denote the fixed bound on the maximum threshold value.
From Lemma \ref{lem:max_threshold}, the two given MSyDSs are inequivalent
iff there is an inequivalence witness with at most $\tau_{max}$ 1's.
In polynomial time, all configurations with at most $\tau_{max}$ 1's can be generated,
and the successor configurations under the two given MSyDSs can be found.

(b) Let $\nu_{max}$ denote the fixed bound on the maximum negative threshold value.
From Lemma \ref{lem:max_negative_threshold}, the two given MSyDSs are inequivalent
iff there is an inequivalence witness with at most $\nu_{max}$ 0's.
In polynomial time, all configurations with at most $\nu_{max}$ 0's can be generated,
and the successor configurations under the two given MSyDSs can be found. 
\end{proof}

The corollary below follows immediately.

\begin{corollary}\label{thm:bounded_node_degree}
If the maximum node degree of any graph in any layer is bounded by a fixed value,
the MSyDS equivalence problem can be solved in polynomial time
when all local functions are threshold functions,
and either every master function is \texttt{OR} or every master function is \texttt{AND}.
\end{corollary}


\section{Expressive Power of MSyDSs}
\label{sec:expressive}

In this section, we consider the expressive power of some classes
of MSyDSs, based on the number of network layers.

\begin{proposition}\label{pro:single_layer}
Every MSyDS is equivalent to a single-layer SyDS.
\end{proposition}

\begin{proof}
Given MSyDS $\mc{S}$, there is an equivalent SyDS $\mc{S'}$ whose
underlying graph is a complete graph.  The local function of a given
node $v$ in  $\mc{S}'$ is the master function of $\mc{S}$ for $v$,
applied to the set of local functions for $v$ on the layers of
$\mc{S}$.  
\end{proof}

\smallskip

\noindent
\textbf{Example:}~ We use the 2-layer system
in Figure~\ref{fig:msyds_example} to illustrate
the construction outlined in the proof of
Proposition~\ref{pro:single_layer}.
The underlying graph of an equivalent single-layer system is a complete graph on 5 nodes.
Let $\text{Thr}(V', q)$, where $V'$ is a subset of
nodes and $q$ is a non-negative integer, denote the
Boolean function which has the value 1 when at least
$q$ of the nodes in $V'$ have the value 1.
Using this notation, the local function associated with
node $v_1$ in the single-layer system is:
\begin{center}
$\text{Thr}(\{v_1, v_3, v_5\}, 2)$ ~\tor{}~
$\text{Thr}(\{v_1, v_2, v_3, v_4\}, 2)$
\end{center}
The local functions for the other nodes of
the system in Figure~\ref{fig:msyds_example} can be
constructed in a similar manner.
\hfill$\Box$

The next result shows that there is a hierarchy of expressive power,
based on the number of layers.

\begin{theorem}\label{thm:layer_hierarchy}
For every $k \geq 2$, there is a MSyDS $\mc{S}_k$ that contains $k$
layers, for which every local function is a threshold function, and
every master function is \texttt{OR}, with the following properties:
For each $\mc{S}_k$, there is no equivalent MSyDS with fewer layers,
for which every local function is symmetric, and every master
function is \texttt{OR}.

However there is an equivalent MSyDS with two layers, for which
every local function is a threshold function, and every master
function is \texttt{AND}.  
\end{theorem} 

\begin{proof}
Let $\mc{S}_k$ be the following $k$ layer MSyDS.  Node set $V$
contains the $k+1$ nodes: $\{ a, b_1, b_2, \, \ldots \, , b_k \}$.
Layer $i$, $1 \leq i \leq k$, contains only one edge: $\{a,b_i\}$.
In every layer, the local function for node $a$ is a threshold
function with threshold 2, and every other local function is the
constant function 0, i.e., a threshold function with threshold value
equal to two plus the node degree in that layer.  Every master
function is \texttt{OR}.

Suppose there exists a MSyDS $\mc{S}'$ with fewer than $k$ layers
that is equivalent to $\mc{S}_k$, such that every local function
of $\mc{S}'$ is symmetric, and every master function is \texttt{OR}.

For convenience, we describe a configuration of $V$ by explicitly
listing the nodes that have value 1 in the configuration.  We use
$\phi$ to denote the configuration with all 0's.

Also, for each layer $q$ of $\mc{S}'$, we let $\hat{f}_{a,q}$ denote
the symmetry table (defined in Section~\ref{sec:prelim}) corresponding
to the symmetric local function $f_{a,q}$.  In other words, if $a$ has degree
$d$ in layer $q$ of $\mc{S}'$, then for each $i$, $0 \leq i \leq
d+1$, $\hat{f}_{a,q}[i]$ is the value of $f_{a,q}$ when there are
exactly $i$ inputs to $f_{a,q}$ that equal 1.

Since $\mc{S}_k$ has the transitions $\phi \, \rightarrow \, \phi$
and $a \, \rightarrow \, \phi$, in every layer $q$ of equivalent
MSyDS $\mc{S}'$, $\hat{f}_{a,q}[0] = 0$ and $\hat{f}_{a,q}[1] = 0$.
For each $i$, $1 \leq i \leq k$, $\mc{S}_k$ has the transition $a
\, b_i \, \rightarrow \, a$, so there is at least one layer of
$\mc{S}'$, say layer $q_i$, such that the graph for layer $q_i$
contains the edge $\{a,b_i\}$, and $\hat{f}_{a,q_i}[2] = 1$.  Since
$\mc{S}'$ has fewer than $k$ layers, there is a layer of $\mc{S}'$,
say layer $q$, such that there are two distinct nodes $b_i$ and
$b_j$ that are both neighbors of $a$ in that layer, and $\hat{f}_{a,q}[2]
= 1$.  Thus, $\mc{S}'$ has the transition $b_i \, b_j \, \rightarrow
\, a$.  However, $\mc{S}_k$ has the transition $b_i \, b_j \,
\rightarrow \, \phi$.  Thus, configuration $b_i \, b_j$ is an
inequivalence witness for $\mc{S}_k$ and $\mc{S}'$.  Consequently,
no such equivalent MSyDS $\mc{S}'$ exists.

However, let $\mc{S}''$ be the following two layer MSyDS with node
set $V$.  Layer 1 contains no edges, and layer 2 contains the $k$
edges $\{a,b_i\}$, $1 \leq i \leq k$.  In layer 1, the local function
for node $a$ is a threshold function with threshold 1.  In layer
2, the local function for node $a$ is a threshold function with
threshold 2.  In both layers, the local function for every node
other than $a$ is the constant function 0.  Every master function
is \texttt{AND}.  It can be verified that $\mc{S}_k$ and $\mc{S}''$
are equivalent. This concludes the proof.
\end{proof}

\section{Additional Remarks and Future Research Directions}
\label{sec:concl}

We studied some fundamental questions for MSyDSs such as the structure
of their phase spaces, the complexity of equivalence problems and
their expressive power.  Researchers have studied several other
fundamental questions for single layer dynamical systems.  Examples
of such questions are \textbf{fixed point existence} (i.e., does
the given single-layer system have a fixed point?) and \textbf{predecessor
existence} (i.e., given a single-layer system and a configuration
\calc{}, is there a configuration \calcp{} such that there is a
one-step transition from \calcp{} to \calc{}?); see
\cite{BH+06,OU-2017,rosenkrantz2024synchronous} and the references
cited therein.  Many of these problems are known to be \cnp-complete
for single-layer systems.  One can easily extend these hardness
results to MSyDSs using the following result.

\begin{proposition}\label{pro:msyds_for_syds}
Let \cals{} be a given single-layer SyDS over the
domain \bbb{} = \{0,1\}.
For any $k \geq 2$, a $k$-layer MsyDS \calsp{} over~ \bbb{}
whose phase space is identical to that of~ \cals{}
can be constructed.
\end{proposition}

\begin{proof}
Let $G(V,E)$ denote the underlying graph of \cals{}.
The MSyDS \calsp{} uses the same set of nodes $V$.
For $1 \leq i \leq k$, the graph $G_i(V, E_i)$ in each layer $i$
is the same as $G$.
Further, in each layer, the local function for a node $v \in V$
is the same as the local function of $v$ in \cals{}.
Finally, the master function for each node $v$ of \calsp{}
is the \tor{} function.
It is straightforward to verify that for each configuration \calc{},
the successor of \calc{} is the same in both \cals{} and \calsp{}.
Thus, the phase spaces of \cals{} and \calsp{} are identical.
\end{proof}

In view of Proposition~\ref{pro:msyds_for_syds}, it is clear that
hardness results for single-layer systems yield
hardness results for multilayer systems as well.

We close by pointing out some directions for future research
concerning computational problems for multi-layer dynamical systems.
First, it will be useful to examine whether other problems that are
efficiently solvable for single-layer systems remain so for multilayer
systems.  For example, it is known that for single-layer SyDSs where
the treewidth of the underlying graph is bounded and the local
functions have certain properties, the existence of of certain
subgraphs of the phase space can be checked
efficiently~\cite{Rosenkrantz-etal-2015}.  It is of interest to
study whether such algorithms can be extended to the multilayer
case.  It is also of interest to investigate whether results for
single-layer SyDSs with stochastic local functions~\cite{Barrett-etal-2011}
can be extended to MSyDSs.  Our work on the expressive power of
MSyDSs considered systems where the local functions are threshold
or symmetric functions.  An interesting direction is to extend the
results on expressive power to multilayer systems with other classes
of Boolean local functions.

\section{Acknowledgments}\label{sec:ack}

We thank the reviewers of AAMAS~2025 and the editors 
for their valuable suggestions.
This work was supported by 
University of Virginia Strategic Investment
Fund award number SIF160, Virginia Department of Health
grant VDH-21-501-0135-1,
DTRA Grant HDTRA1-24-R-0028,
and NSF Grants CCF-1918656 (Expeditions) and
OAC-1916805 (CINES).
Opinions, findings, and conclusions are
those of the authors and do not necessarily reflect the view of the
funding entities.



\balance

\bibliographystyle{ACM-Reference-Format} 
\bibliography{ref_one,ref_two}


\end{document}